\documentclass[12pt,twoside]{article}

\usepackage[a4paper]{geometry}
\makeatletter
\renewcommand\title[1]{\gdef\@title{\reset@font\Large\bfseries #1}}
\renewcommand\section{\@startsection {section}{1}{\z@}%
                                   {-3.5ex \@plus -1ex \@minus -.2ex}%
                                   {2.3ex \@plus.2ex}%
                                   {\normalfont\large\bfseries}}
\renewcommand\subsection{\@startsection{subsection}{2}{\z@}%
                                     {-3ex\@plus -1ex \@minus -.2ex}%
                                     {1.5ex \@plus .2ex}%
                                     {\normalfont\normalsize\bfseries}}
\renewcommand\subsubsection{\@startsection{subsubsection}{3}{\z@}%
                                     {-2.5ex\@plus -1ex \@minus -.2ex}%
                                     {1.5ex \@plus .2ex}%
                                     {\normalfont\normalsize\bfseries}}

\def\@runningauthor{}\newcommand{\runningauthor}[1]{\def\runningauthor{#1}}
\def\@runningtitle{}\newcommand{\runningtitle}[1]{\def\runningtitle{#1}}

\renewcommand{\ps@plain}{%
\renewcommand{\@evenfoot}{\footnotesize Sequences and Their Applications (SETA) 2018\hfill\thepage}
\renewcommand{\@oddfoot}{\footnotesize Sequences and Their Applications (SETA) 2018 \hfill\thepage}
\renewcommand{\@evenhead}{\footnotesize\scshape \hfill\runningauthor\hfill}
\renewcommand{\@oddhead}{\footnotesize\scshape \hfill\runningtitle\hfill}}
\g@addto@macro\bfseries{\boldmath}

\makeatother

\usepackage{amsthm,amsmath,amssymb}

\usepackage{graphicx}

\usepackage[colorlinks=true,citecolor=black,linkcolor=black,urlcolor=blue]{hyperref}

\theoremstyle{plain}
\newtheorem{theorem}{Theorem}
\newtheorem{lemma}[theorem]{Lemma}

\theoremstyle{definition}

\newtheorem{example}[theorem]{Example}

\theoremstyle{remark}

\title{The 2-adic complexity of a class of binary sequences with optimal autocorrelation magnitude}

\runningtitle{The 2-adic complexity of a class of binary sequences with optimal autocorrelation magnitude}

\author{Yuhua Sun$^{a,c}$\thanks{Yuhua Sun is financially supported by Shandong Provincial Natural Science Foundation of China (No. ZR2017MA001, No. ZR2016FL01), the Open Research Fund from Shandong provincial Key Laboratory of Computer Networks, Grant No. SDKLCN-2017-03.},
\quad
Tongjiang Yan$^{a,b}$\thanks{Tongjiang Yan is the corresponding author and is financially supported by Qingdao application research on special independent innovation plan project (No.16-5-1-5-jch), the Open Research Fund from Key Laboratory of Applied Mathematics of Fujian Province University (Putian University) (No.SX201702), and the Fundamental Research Funds for the Central Universities (No.17CX02030A).
}, \quad Zhixiong Chen$^{b}$\\
\small a. College of Sciences, China University of Petroleum,\\
\small Qingdao 266555,
Shandong, China\\
\small b. Provincial Key Laboratory of Applied Mathematics, Putian University, \\
\small Putian, Fujian 351100, P.R. China\\
\small c. Qilu University of Technology (Shandong Academy of Sciences), Shandong Computer \\
\small Science Center (National Supercomputer Center in Jinan), \\
\small Shandong Provincial Key Laboratory of Computer Networks.\\
\tt $\{$sunyuhua\_1,yantoji$\}$@163.com; ptczx@126.com
}

\runningauthor{Y.\ Sun, T.\ Yan, Z.\ Chen}

\date{}

\begin{document}

\maketitle

\thispagestyle{empty}

\begin{abstract}
Recently, a class of binary sequences with optimal autocorrelation magnitude has been presented by Su et al. based on interleaving technique and Ding-Helleseth-Lam sequences (Des. Codes Cryptogr., https://doi.org/10.1007/s10623-017-0398-5). And its linear complexity has been proved to be large enough to resist the B-M Algorighm (BMA) by Fan (Des. Codes Cryptogr., https://doi.org/10.1007/s10623-018-0456-7). In this paper, we study the 2-adic complexity of this class of binary sequences.
Our result shows that the 2-adic complexity of this class of sequence is no less than one half of its period, i.e., its 2-adic complexity is large enough to resist the Rational Aproximation Algorithm (RAA).
\end{abstract}


\section{Introduction}

Let $v$ be a positive integer and let $s_{i}=(s_{i}(0),s_{i}(1),\cdots,s_{i}(v-1))$, $0\leq i\leq u-1$, be $u$ binary sequences of period $v$. Based on these $u$ binary sequences, a $v\times u$ matrix $I=(I_{i,j})$ is given by
\begin{equation}
I=\left(
\begin{array}{cccc}
s_{0}(0) & s_{1}(0) &\cdots &s_{u-1}(0)\\
s_{0}(1) & s_{1}(1) &\cdots &s_{u-1}(1)\\
\vdots & \vdots &\ddots &\vdots\\
s_{0}(v-1) & s_{1}(v-1) &\cdots &s_{u-1}(v-1)
\end{array}
\right).
\end{equation}
Then an interleaved sequence $s$ can be obtained by concatenating the successive rows of the above matrix $I$. For convenience, the sequence $s$ is written  as
$$
s=I(s_{0},s_{1},\cdots,s_{u-1}),
$$
where $I$ represents the interleaved operator. This construction method of interleaved sequences, called interleaving technique, was originally presented by Gong \cite{Gong}. In order to give consideration to both long period and low autocorrelation, most related literatures focus on the case of $u=4$. Note that the period $N$ of an interleaved sequence satisfies $N=4v\equiv0\ \mathrm{mod}\ 4$ for $u=4$. For any binary sequence $s$ with period $N\equiv0\ \mathrm{mod}\ 4$, it is called optimal in the sense of autocorrelation value if its autocorrelation value $AC_s(\tau)\in\{N,0,4\}$ or $\{N,0,-4\}$, and called optimal in the sense of autocorrelation magnitude if $AC_s(\tau)\in\{N,0,4,-4\}$.

Binary sequences with optimal autocorrelation value/magnitude, large linear complexity and high 2-adic complexity have important applications in stream cipher. Therefore, constructions and analyses of sequences meeting all the three properties have become research hotspots for pseudo-random sequence investigators.

At present, most of the known interleaved sequences of $v\times4$ are with the form $s=I(s_{0},s_{1},s_{0}^{\prime},s_{1}^{\prime})$, $s=I(s_{0},s_{0}^{\prime},s_{1},s_{1}^{\prime})$ or $s=I(s_{0},s_{1},s_{2},s_{2}^{\prime})$, where $s_{0}^{\prime}$, $s_{1}^{\prime}$ and $s_{2}^{\prime}$ are shift sequences, or complementary sequences or complementary sequences of shift sequences of sequences $s_{0}$, $s_{1}$, $s_{2}$ respectively. Sequences with the previous two forms are essentially constructed by two different sequences and sequences with the third form are constructed by three diferent sequences.
For example, using interleaving technique, Arasu, Ding and Helleseth et al. \cite{Arasu} constructed a class of sequences with optimal autocorrelation value, which later were proved to be with large linear complexity by Wang and Du \cite{Wang-Du}. Then, Tang and Gong \cite{Tang-Gong} constructed three classes of sequences with optimal autocorrelation value/magnitude and these sequences have also been verified to be with large linear complexity by Li and Tang \cite{Li Nian}. Next, Tang and Ding \cite{Tang-Ding} gave a more general construction which consists of the constructions of both \cite{Arasu} and \cite{Tang-Gong}. Moreover, Yan et al. \cite{Yan-Information S} also generalized the construction of \cite{Tang-Gong} and they also presented a sufficient and necessary condition for an interleaved sequence to be with optimal autocorrelation. It should be pointed out that all the interleaved sequences mentioned above were essentially constructed by two different sequences, i.e., they have the form $s=I(s_{0},s_{1},s_{0}^{\prime},s_{1}^{\prime})$ or $s=I(s_{0},s_{0}^{\prime},s_{1},s_{1}^{\prime})$. But there were no results about the 2-adic complexities of these interleaved binary sequences until Xiong et al. \cite{Xiong Hai-1} presented a method of circulant matrix and Hu \cite{Hu Honggang} gave a method of using exact autocorrelation value distributions to determine the 2-adic complexities of binary sequences. Using circulant matrix, Xiong et al. proved that all sequences with ideal autocorrelation and two classes of sequences with optimal autocorrelation, i.e., Legendre sequences and Ding-Helleseth-Lam sequences \cite{Ding-H-L}, have maximal 2-adic complexities \cite{Xiong Hai-1}. They also show that all the interleaved sequences with optimal autocorrelation value/magnitude mentioned above have also maximal 2-adic complexities. Interesting enough, most of the results in \cite{Xiong Hai-1,Xiong Hai-2} can also be simply proved by Hu's method given in \cite{Hu Honggang}.

Very recently, using three different sequences and the above third interleaved structure $s=I(s_{0},s_{1},s_{2},s_{2}^{\prime})$, Su et al. \cite{Su wei} constructed a class of binary sequences with optimal autocorrelation using interleaving technique and Ding-Helleseth-Lam sequences. In quick succession, Fan \cite{Fan Cuiling} proved that this class of sequences has large linear complexity.

In this paper, combining the methods of Hu \cite{Hu Honggang} and Xiong et al. \cite{Xiong Hai-2}, we will study the 2-adic complexity of this new class of sequence. A lower bound on the 2-adic complexity will be given in Section 2, and we will give a summary in Section 3. Our result shows that the 2-adic complexity of this new class of  sequences is at least one half of its period, which is large enough to resist against the Rational Algorithm \cite{Andrew Klapper}.


\section{Main result}\label{section 2}

In this section, we will study the 2-adic complexities of the sequences constructed by Su et al. \cite{Su wei} and give lower bounds on the 2-adic complexities. Before presenting our main result, we first give the definition of 2-adic complexity of a binary sequence and a simple description of the construction by Su et al..

Let $N$ be a positive integer and $\mathbf{Z}_{N}$ the ring of integers modulo $N$. And let $s=\Big(s(0),s(1),\cdots,s(N-1)\Big)$ be a binary sequence of period $N$.
Denote $S(x)=\sum\limits_{i=0}^{N-1}s(i)x^i\in \mathbb{Z} [x]$ and suppose
\begin{equation}
\frac{S(2)}{2^N-1}=\frac{\sum\limits_{i=0}^{N-1}s(i)2^i}{2^N-1}=\frac{e}{f},\ 0\leq e\leq f,\ \mathrm{gcd}(e,f)=1.\nonumber
\end{equation}
Then the integer $\lfloor\mathrm{log}_2(f+1)\rfloor$ is called the 2-adic complexity of the sequence $s$ and is denoted as $\Phi_{2}(s)$, i.e.,
\begin{equation}
\Phi_{2}(s)=\left\lfloor\mathrm{log}_2\left(\frac{2^N-1}{\mathrm{gcd}\left(2^N-1,S(2)\right)}+1\right)\right\rfloor,\label{2-adic calculation}
\end{equation}
where $\lfloor z\rfloor$ is the smallest integer that is greater than or equal to $z$.
For a binary periodic sequence used as the key stream generator in a stream cipher, its 2-adic complexity should be no less than one half of its period to resist the RAA (Rational Approximation Algorithm) \cite{Andrew Klapper}.

Denote
$C_{s}=\{i\in \mathbf{Z}_N|s(i)=1\}$. Then we call $C_{s}$ the support of the sequence $s$.
Let $p$ be an odd prime satisfying $p=4k+1=a^{2}+4b^{2}$, where $k$ is odd integer and $b=\pm1$. Suppose $g$ is a primitive root of $p$.
Define
\begin{eqnarray}
D_j&=&\{g^{4i+j}\ \mathrm{mod}\ p\mid i=0, 1, \cdots, k-1\},\ j=0, 1, 2, 3.\nonumber
\end{eqnarray}
Then each $D_j$ is called the cyclotomic class of order $4$ with respect to $\mathbf{Z}_{p}$, where $0\leq j\leq3$.
Let $s^{(1)},\ s^{(2)},\ s^{(3)},\ s^{(4)}$ be the binary sequences of period $p$ with supports $D_{0}\cup D_{1},\ D_{0}\cup D_{3},\ D_{1}\cup D_{2},\ D_{2}\cup D_{3}$, respectively. Then each $s^{i}$ is called Ding-Helleseth-Lam sequence, which has been proved to be optimal sequence with autocorrelation values $1$ and $-3$ by Ding et al. \cite{Ding-H-L}.
Suppose that binary sequence $w=(w(0),w(1),w(2),w(3))$ satisfy $w(0)=w(2)$, $w(1)=w(3)$, and $d$ satisfies $4d\equiv1\ \mathrm{mod}\ p$. Then one of the interleaved sequences by Su et al. is given by
\begin{eqnarray}
s=I\Big(s^{(3)}+w(0), \mathbf{L}^{d}(s^{(2)})+w(1),\mathbf{L}^{2d}(s^{(1)})+w(2),\mathbf{L}^{3d}(s^{(1)})+w(3)\Big),\label{sequence-de}
\end{eqnarray}
where $s^{\prime}+1=\Big(s^{\prime}(0)+1,s^{\prime}(1)+1,\cdots,s^{\prime}(p-1)+1\Big)$ and  $\mathbf{L}^{d}(s^{\prime})=\Big(s^{\prime}(d),s^{\prime}(d+1),\cdots,s^{\prime}(p-1),s^{\prime}(0),s^{\prime}(1),s^{\prime}(d-1)\Big)$ is the cyclic left shift operator of length $d$ for a binary sequence $s^{\prime}=\Big(s^{\prime}(0),s^{\prime}(1),\cdots,s^{\prime}(p-1)\Big)$ of period $p$.

Now, we present our main result.
\begin{theorem}\label{main result}
Let $s=I\Big(s^{(3)}+w(0), \mathbf{L}^{d}(s^{(2)})+w(1),\mathbf{L}^{2d}(s^{(1)})+w(2),\mathbf{L}^{3d}(s^{(1)})+w(3)\Big)=\Big(s(0),s(1),\cdots,s(4p-1)\Big)$ be the sequence of period $N=4p$ defined in Eq. (\ref{sequence-de}) by Su et al.. Then the 2-adic complexity $\Phi(s)$ of the sequence $s$ satisfy the following lower bound
\begin{eqnarray}
2p\leq\Phi(s)\leq 4p-2,
\end{eqnarray}
i.e., the 2-adic complexity of $s$ is no less than one half of the period.
\end{theorem}

In order to prove Theorem \ref{main result}, we need a series of lemmas. Firstly, we list the autocorrelation distribution of the sequence $s$ in Eq. (\ref{sequence-de}) which has been showed by Su et al. \cite{Su wei}.
\begin{lemma} \label{O-A-M}(\cite{Su wei})
Let $s=I\Big(s^{(3)}+w(0), \mathbf{L}^{d}(s^{(2)})+w(1),\mathbf{L}^{2d}(s^{(1)})+w(2),\mathbf{L}^{3d}(s^{(1)})+w(3)\Big)$ be the same as that in Eq. (\ref{sequence-de}). Then $s$ is optimal with respect to the autocorrelation magnitude. Specifically, for any $\tau$, $1\leq\tau\leq 4p-1$, if we write $\tau=\tau_1+4\tau_2$, where ($\tau_1=0$ and $1\leq\tau_2\leq p-1$) or ($1\leq\tau_1\leq3$ and $0\leq\tau_2\leq p-1$), the distribution of the autocorrelation is given by
$$
AC_s(\tau)=\left\{
\begin{array}{lllllllllll}
-4,\ \ \ \ \ \tau_1=0,\ \tau_2\neq0,\\
-4,\ \ \ \ \ \tau_1=1,\ \tau_2+d=0\ (\mathrm{mod}\ p),\\
-4b,\ \ \ \ \tau_1=1,\ (\tau_2+d)\ (\mathrm{mod}\ p)\in D_{0}\cup D_{2},\\
4b,\ \ \ \ \ \ \tau_1=1,\ (\tau_2+d)\ (\mathrm{mod}\ p)\in D_{1}\cup D_{3},\\
4,\ \ \ \ \ \ \ \tau_1=2,\ \tau_2+2d=0\ (\mathrm{mod}\ p),\\
0,\ \ \ \ \ \ \ \tau_1=2,\ (\tau_2+2d)\neq0\ (\mathrm{mod}\ p),\\
-4,\ \ \ \ \ \tau_1=3,\ (\tau_2+3d)=0\ (\mathrm{mod}\ p),\\
4b,\ \ \ \ \ \ \tau_1=3,\ (\tau_2+3d)\ (\mathrm{mod}\ p)\in D_{0}\cup D_{2},\\
-4b,\ \ \ \ \tau_1=3,\ (\tau_2+3d)\ (\mathrm{mod}\ p)\in D_{1}\cup D_{3}.
\end{array}
\right.
$$
\end{lemma}
In this paper, the method of Hu \cite{Hu Honggang} will be employed to analyze the 2-adic complexity of $s$ and it can be described as the following Lemma \ref{Sun} which has also been used in our another work \cite{Sun Yuhua-17}.
\begin{lemma}\label{Sun}\cite{Hu Honggang,Sun Yuhua-17}
Let $s=\Big(s(0),s(1),\cdots,s(N-1)\Big)$ be a binary sequence of period $N$, $S(x)=\sum\limits_{i=0}^{N-1}s(i)x^i\in \mathbb{Z}[x]$ and $T(x)=\sum\limits_{i=0}^{N-1}(-1)^{s(i)}x^i\in \mathbb{Z}[x]$. Then
\begin{eqnarray}
-2S(x)T(x^{-1})\equiv N+\sum\limits_{\tau=1}^{N-1}AC(\tau)x^{\tau}-T(x^{-1})\left(\sum\limits_{i=0}^{N-1}x^i\right)\pmod{x^N-1}.\nonumber
\end{eqnarray}
\end{lemma}

\begin{lemma}\label{main-lemma}
Let $s=\Big(s(0),s(1),\cdots,s(4p-1)\Big)$ be the sequence of period $N=4p$ defined in Eq. (\ref{sequence-de}), $S(x)=\sum_{i=0}^{4p-1}s(i)x^{i}\in \mathbb{Z}[x]$ and $T(x)=\sum_{i=0}^{4p-1}(-1)^{s(i)}x^{i}\in \mathbb{Z}[x]$, where $4d\equiv1\ \mathrm{mod}\ p$. Then
\begin{eqnarray}
S(2)T(2^{-1})\equiv && 2\Big[\frac{2^{4p}-1}{2^4-1}+(2^{2p}+1)(2^{p}-1)\nonumber\\
&&-2^{p}(2^{2p}-1)b\sum_{i\in \mathbf{Z}_{p}^{\ast}}(\frac{i}{p})2^{4i}-p  \Big]\pmod{2^{4p}-1},\nonumber
\end{eqnarray}
where $\big(\frac{i}{p}\big)$ is the Legendre symbol.
\end{lemma}
\begin{proof}
Above all, by Lemma \ref{Sun}, we know that
\begin{eqnarray}
S(2)T(2^{-1})&&\equiv-2p-\frac{1}{2}\sum_{\tau=1}^{4p-1}AC_s(\tau)2^{\tau}\ (\mathrm{mod}\ 2^{4p}-1)\label{main-conc}
\end{eqnarray}
Note that $d=\frac{3p+1}{4}$ for $p=4k+1$ and $4d\equiv1\ \mathrm{mod}\ p$. Then, for $0\leq\tau_2\leq p-1$, we have
\begin{eqnarray}
&&\tau_2+d=0\ (\mathrm{mod}\ p)\Leftrightarrow\tau_2=\frac{p-1}{4}\ \mathrm{mod}\ p,\label{trans-1}\\
&&\tau_2+2d=0\ (\mathrm{mod}\ p)\Leftrightarrow\tau_2=\frac{p-1}{2}\ \mathrm{mod}\ p,\label{trans-2}\\
&&\tau_2+3d=0\ (\mathrm{mod}\ p)\Leftrightarrow\tau_2=\frac{3(p-1)}{4}\ \mathrm{mod}\ p,\label{trans-3}\\
&&(\tau_2+d)\ (\mathrm{mod}\ p)\in D_{0}\cup D_{2}\Leftrightarrow\tau_2\in\Big(D_{0}\cup D_{2}+\frac{p-1}{4}\Big),\label{trans-4}\\
&&(\tau_2+2d)\ (\mathrm{mod}\ p)\in D_{0}\cup D_{2}\Leftrightarrow\tau_2\in\Big(D_{0}\cup D_{2}+\frac{p-1}{2}\Big),\label{trans-5}\\
&&(\tau_2+3d)\ (\mathrm{mod}\ p)\in D_{0}\cup D_{2}\Leftrightarrow\tau_2\in\Big(D_{0}\cup D_{2}+\frac{3(p-1)}{4}\Big),\label{trans-6}\\
&&(\tau_2+d)\ (\mathrm{mod}\ p)\in D_{1}\cup D_{3}\Leftrightarrow\tau_2\in\Big(D_{1}\cup D_{3}+\frac{p-1}{4}\Big),\label{trans-7}\\
&&(\tau_2+2d)\ (\mathrm{mod}\ p)\in D_{1}\cup D_{3}\Leftrightarrow\tau_2\in\Big(D_{1}\cup D_{3}+\frac{p-1}{2}\Big),\label{trans-8}\\
&&(\tau_2+3d)\ (\mathrm{mod}\ p)\in D_{1}\cup D_{3}\Leftrightarrow\tau_2\in\Big(D_{1}\cup D_{3}+\frac{3(p-1)}{4}\Big),\label{trans-9}
\end{eqnarray}
where $D_{0}\cup D_{2}+x=\{a+x\ \mathrm{mod}\ p|a\in D_0\cup D_2\}$ and $D_{1}\cup D_{3}+x=\{a+x\ \mathrm{mod}\ p|a\in D_1\cup D_3\}$.
By Lemma \ref{O-A-M} and Eqs. (\ref{trans-1}-\ref{trans-9}), we know that
\begin{eqnarray}
\sum_{\tau=1}^{4p-1}AC_s(\tau)2^{\tau}&&=\sum_{\tau_{2}=1}^{p-1}AC_s(4\tau_{2})2^{4\tau_{2}}+\sum_{\tau_{1}=1}^{3}\sum_{\tau_{2}=0}^{p-1}AC_s(\tau_1+4\tau_{2})2^{\tau_1+4\tau_{2}}\nonumber\\
&&=4-4\times\sum_{\tau_{2}=0}^{p-1}2^{4\tau_{2}}+\sum_{\tau_{1}=1}^{3}\sum_{\tau_{2}=0}^{p-1}AC_s(\tau_1+4\tau_{2})2^{\tau_1+4\tau_{2}}\nonumber\\
&&=4-4\times\frac{2^{4p}-1}{2^{4}-1}-4\times2^{1+4\times\frac{p-1}{4}}\nonumber\\
&&+4\times2^{2+4\times\frac{p-1}{2}}+0-4\times2^{3+4\times\frac{3(p-1)}{4}}\nonumber\\
&&-4b\sum_{\tau_2\in(D_0\cup D_2+\frac{p-1}{4})}2^{1+4\tau_2}+4b\sum_{\tau_2\in(D_1\cup D_3+\frac{p-1}{4})}2^{1+4\tau_2}\nonumber\\
&&+4b\sum_{\tau_2\in(D_0\cup D_2+\frac{3(p-1)}{4})}2^{3+4\tau_2}-4b\sum_{\tau_2\in(D_1\cup D_3+\frac{3(p-1)}{4})}2^{3+4\tau_2}\nonumber\\
&&\equiv4-4\times\frac{2^{4p}-1}{2^{4}-1}-4\times2^{p}+4\times2^{2p}-4\times2^{3p}\nonumber\\
&&-4b\times2^p\Big(\sum_{\tau_2^{\prime}\in D_0\cup D_2}2^{4\tau_2^{\prime}}-\sum_{\tau_2^{\prime}\in D_1\cup D_3}2^{4\tau_2^{\prime}}\Big)\nonumber\\
&&+4b\times2^{3p}\Big(\sum_{\tau_2^{\prime}\in D_0\cup D_2}2^{4\tau_2^{\prime}}-\sum_{\tau_2^{\prime}\in D_1\cup D_3}2^{4\tau_2^{\prime}}\Big)\nonumber\\
&&\equiv-4\Big[\frac{2^{4p}-1}{2^{4}-1}+(1+2^{2p})(2^{p}-1)\nonumber\\
&&-2^p(2^{2p}-1)b\sum_{\tau_2^{\prime}\in Z_p^{\ast}}\big(\frac{\tau_2}{p}\big)2^{4\tau_2^{\prime}}\Big](\mathrm{mod}\ 2^{4p}-1).\label{Auto-sum}
\end{eqnarray}
Substituting the Eq. (\ref{Auto-sum}) into the Eq. (\ref{main-conc}), one can obtain the result.
\end{proof}

\begin{lemma}\label{main-lemma-3}
Let the symbols be the same as those in Lemma \ref{main-lemma}. Then we have
(i) $\mathrm{gcd}(S(2),3)=1$,\ $\mathrm{gcd}(S(2),5)=5$; (ii) $3|2^{2p}-1,\ 5|2^{2p}+1$.
\end{lemma}
\begin{proof} (i) Without loss of generality, let $w(0)=w(2)=0$ and $w(1)=w(3)=1$. Since $d=\frac{3p+1}{4}$, we have $-d=\frac{p-1}{4}\ \mathrm{mod}\ p$. Then
\begin{eqnarray}
S(2)=&&\sum_{i\in D_{1}\cup D_2}2^{4i}+\sum_{i\in \{\frac{p-1}{4}\}\cup(D_{1}\cup D_2+\frac{p-1}{4})}2^{4i+1}\nonumber\\
&&+\sum_{i\in (D_{0}\cup D_1+\frac{p-1}{2})}2^{4i+2}+\sum_{i\in \{\frac{3(p-1)}{4}\}\cup(D_{2}\cup D_3+\frac{3(p-1)}{4})}2^{4i+3}\nonumber\\
\equiv&&\sum_{i\in D_{1}\cup D_2}2^{4i}+2^p+2^p\sum_{i\in D_{1}\cup D_2}2^{4i}\nonumber\\
&&+2^{2p}\sum_{i\in (D_{0}\cup D_1)}2^{4i}+2^{3p}+2^{3p}\sum_{i\in D_{2}\cup D_3}2^{4i}(\mathrm{mod}\ 2^{4p}-1).\nonumber
\end{eqnarray}
Note that $2^p\equiv-1\ \mathrm{mod}\ 3$, $2^4=16\equiv1\ \mathrm{mod}\ 3$, $2^p\equiv2\ \mathrm{mod}\ 5$, and $2^4=16\equiv1\ \mathrm{mod}\ 5$. We have
\begin{eqnarray}
S(2)\equiv&&\frac{p-1}{2}-1-\frac{p-1}{2}+\frac{p-1}{2}-1-\frac{p-1}{2}\equiv-2\ (\mathrm{mod}\ 3),\nonumber\\
S(2)\equiv&&\frac{p-1}{2}+2+2\times\frac{p-1}{2}-\frac{p-1}{2}+3+3\times\frac{p-1}{2}\equiv0\ (\mathrm{mod}\ 5).\nonumber
\end{eqnarray}
The results follow. \\
(ii) Note that $2^{2p}-1=4^p-1=(3+1)^p-1=0\ (\mathrm{mod}\ 3)$ and $2^{2p}+1=(5-1)^p+1\equiv0\ (\mathrm{mod}\ 5)$. The results hold.
\end{proof}

\begin{lemma}\label{main-lemma-2}
Let $p$ be an odd prime. Then we have $\mathrm{gcd}(p,2^p-1)=1$ and $\mathrm{gcd}(p+4,\frac{2^p+1}{3})=1$.
\end{lemma}

\begin{proof} Above all, we will prove $\mathrm{gcd}(p,2^p-1)=1$. Suppose $\mathrm{gcd}(p,2^p-1)>1$. Then $p|2^p-1$. But we know that $p|2^{p-1}-1$ by Fermat Theorem and $\mathrm{gcd}(2^p-1,2^{p-1}-1=2^{\mathrm{gcd}(p,p-1)}-1=1$ which is a contradiction to $\mathrm{gcd}(p,2^p-1)>1$.

Next, we will prove $\mathrm{gcd}(p+4,\frac{2^p+1}{3})=1$. It is easy to verify that this holds for $p=3$. Suppose $p\geq5$ and $r>1$ is a common prime factor of $p+4$ and $\frac{2^p+1}{3}$. Then $r\leq p+4<2p$. Let $\mathrm{ord}_r(2)$ be the multiplicative order of 2 modulo $r$. Then $r|2^{2p}-1$ which implies $\mathrm{ord}_r(2)|2p$. Therefore, $\mathrm{ord}_r(2)=2$, $p$ or $2p$. In the following, we consider the three cases respectively.
\begin{itemize}
\item[(1)]$\mathrm{ord}_r(2)=2$. In this case, it is obvious that $r=3$. However, we will prove that $3\nmid\frac{2^p+1}{3}$, i.e., $9\nmid2^p+1$. Note that $p\geq5$ is a prime. Then $p=6t+1$ or $6t+5$ for some non-negative integer $t$. Without loss of generality, let $p=6t+1$. Then $2^p+1=2^{6t+1}+1=2\times2^{6t}+1\equiv3\pmod{9}$ which implies that $9\nmid 2^p+1$. Similarly, $9\nmid 2^p+1$ for $p=6t+5$. This contradict to that $r$ is a common prime factor of $p+4$ and $\frac{2^p+1}{3}$.\\
\item[(2)]$\mathrm{ord}_r(2)=p$. In this case, we have $r|2^p-1$ which contradict to $r|\frac{2^p+1}{3}$.\\
\item[(3)]$\mathrm{ord}_r(2)=2p$. Note that $r|2^{r-1}-1$ by Fermat Theorem. Then we have $r-1\geq2p$, i.e., $r\geq2p+1$ which contradict to $r\leq p+4<2p$.
\end{itemize}
The result follows.
\end{proof}

Combining the above Lemmas \ref{main-lemma}-\ref{main-lemma-2}, we can prove our main result.

\begin{proof}[Proof of Theorem~\ref{main result}] Firstly, by Lemma \ref{main-lemma}, we have
\begin{eqnarray}
S(2)T(2^{-1})\equiv 2(2^p-1)-p\pmod{\frac{2^{2p}-1}{3}}.
\end{eqnarray}
Note that $3|2^p+1$ and $3\nmid 2^p-1$. Then, Furthermore, we have
\begin{eqnarray}
S(2)T(2^{-1})\equiv -p\pmod{2^p-1}
\end{eqnarray}
and
\begin{eqnarray}
S(2)T(2^{-1})\equiv -(p+4)\pmod{\frac{2^p+1}{3}}.
\end{eqnarray}
By Lemma \ref{main-lemma-2}, we know that
\begin{eqnarray}
\mathrm{gcd}(-p,2^p-1)=\mathrm{gcd}(p,2^p-1)=1
\end{eqnarray}
and
\begin{eqnarray}
\mathrm{gcd}\Big(-(p+4),\frac{2^p+1}{3}\Big)=\mathrm{gcd}\Big(p+4,\frac{2^p+1}{3}\Big)=1.
\end{eqnarray}
Note that $\mathrm{gcd}\Big(S(2),\frac{2^{2p}-1}{3}\Big)\leq \mathrm{gcd}\Big(S(2)T(2^{-1},\frac{2^{2p}-1}{3}\Big)=1$. By Lemma \ref{main-lemma-3}, we know $\mathrm{gcd}\Big(S(2),3\Big)=1$. Then we have $\mathrm{gcd}\Big(S(2),2^{2p}-1\Big)=1$. Therefore, we get $\frac{2^{4p}-1}{\mathrm{gcd}\Big(S(2),2^{4p}-1\Big)}\geq 2^{2p}-1$.

Finally, from Lemma \ref{main-lemma-3}, we have $5|\mathrm{gcd}\Big(S(2),2^{4p}-1\Big).$
By the Eq. (\ref{2-adic calculation}), the result follows.
\end{proof}
\begin{example}
Let $p=13$ and $N=4p=52$. By a Mathematica Program, we get the corresponding $\mathrm{gcd}\big(S(2), 2^{52}-1)=5$ and $\Phi(s)=49$, i.e., the 2-adic complexity arrives almost the upper bound in Theorem 1.
\end{example}

\section{Summary}
In this paper, we study the 2-adic complexity of a class of binary sequences with optimal autocorrelation magnitude and our result shows
that the 2-adic complexity of this class of sequences is no less than the half of its period which implies that these sequences can resist
 the RAA for FCSRs.

In fact, several classes of binary sequences with optimal autocorrelation magnitude have been obtained from Construction 1 in \cite{Su wei} by using different $(s_{0},s_{1},s_{2},s_{3})$ and $w=(w(0),w(1),w(2),w(3))$ and we can also show that the corresponding 2-adic complexity of each class of these sequence is no less than the half of its period, where $w(0)=w(2)$ and $w(1)=w(3)$.

Finally, it seems that the 2-adic complexity of sequences constructed by Su et al. from many examples could arrive the upper bound in Theorem 1. But, due to our limited ability, we can not prove $\mathrm{gcd}(S(2),2^{2p}+1)=5$ in this paper. Therefore, we sincerely invite those readers who are interested in this problem to participate in this work.




\end{document}